\documentclass[12pt]{article}
\usepackage{natbib}
\usepackage{amsmath,amssymb,amsthm,graphicx}
\usepackage{rotating}
\usepackage{longtable}
\usepackage{enumerate}
\usepackage{rotating}
\usepackage{epsfig}
\usepackage{subfigure}
\usepackage{float}
\usepackage{color}
\input epsf

\theoremstyle{plain}
\newtheorem{thm}{Theorem}
\newtheorem{lemma}{Lemma}
\newtheorem{prop}{Proposition}

\theoremstyle{definition}
\newtheorem*{defn}{Definition}

\newcommand{\zz}{\mathbb{Z}}
\newcommand{\calb}{\mathcal{B}}
\newcommand{\calm}{\mathcal{M}}
\newcommand{\cals}{\mathcal{S}}
\newcommand{\calt}{\mathcal{T}}
\newcommand{\diag}{\operatorname{diag}}
\renewcommand{\flat}{\operatorname{Flat}}

\title{When Do Phylogenetic Mixture Models Mimic Other Phylogenetic Models?}

\begin{document}
\bibliographystyle{sysbio}

\begin{center}{\LARGE When Do Phylogenetic Mixture Models Mimic Other Phylogenetic Models?}\end{center} 

\hspace{0.2in}

\begin{center}
{Elizabeth S. Allman$^{1}$, John A. Rhodes$^{1}$, Seth Sullivant$^{2}$\\

\hspace{0.2in}

\normalsize{$^{1}$\itshape{Department of Mathematics and Statistics, University of Alaska Fairbanks, Box 756660, Fairbanks, AK, 99775}}\\
\normalsize{$^{2}$\itshape{Department of Mathematics, North Carolina State University, Box 8205
     Raleigh, NC 27695 }}\\

}
\end{center}

\vspace{0.8in}

\makeatletter
\renewcommand{\section}{\@startsection
{section}%
{1}%
{0mm}%
{-\baselineskip}%
{\baselineskip}%
{\normalfont\scshape\centering}}

\makeatletter
\renewcommand{\subsection}{\@startsection
{subsection}%
{1}%
{0mm}%
{-\baselineskip}%
{\baselineskip}%
{\normalfont\itshape\centering}}

\makeatletter
\renewcommand{\subsubsection}{\@startsection
{subsubsection}%
{1}%
{0mm}%
{-\baselineskip}%
{-1.0\baselineskip}%
{\normalfont\itshape}}

{\bfseries 
Abstract--}
Phylogenetic mixture models, in which the sites in sequences undergo different substitution processes along the same or different trees, allow the description of heterogeneous evolutionary processes. As data sets consisting of longer sequences become available, it is important to understand such models, for both theoretical insights and use in statistical analyses.  Some recent articles have highlighted disturbing ``mimicking'' behavior in which a distribution from a mixture model is identical to one arising on a different tree or trees. Other works have indicated such problems are unlikely to occur in practice, as they require very special parameter choices.

After surveying some of these works on mixture models, we give several new results.  In general, if the number of components in a generating mixture is not too large and we disallow zero or infinite branch lengths, then it cannot mimic the behavior of a non-mixture on a different tree. On the other hand, if the mixture model is locally over-parameterized, it is possible for a phylogenetic mixture model to mimic distributions of another tree model. Though theoretical questions remain, these sorts of results can serve as a guide to when the use of mixture models in either ML or Bayesian frameworks is likely to lead to statistically consistent inference, and when mimicking due to heterogeneity should be considered a realistic possibility.

\ 

\newpage

\section*{}

As phylogenetic models have developed, there has been a trend toward allowing increasing heterogeneity of the evolutionary processes from site to site. For instance, the standard general time-reversible model (GTR) is now usually augmented by across-site rate variation, and the inclusion of invariable sites. Recently interest has expanded to more general mixture models, in which processes vary more widely. Much of this work has focused on developing models that might be useful for data analysis, and has therefore involved gaining practical experience with inference from data sets, and investigating theoretical questions of parameter identifiability, which is necessary for establishing that inference is  statistically consistent.

Among the results emerging from theoretical considerations, however, has been the construction of some explicit examples of mixture models  on one tree that `mimic' standard models on another tree, for certain parameter choices 
\citep{Matsen2007,StefVig2007,StefVig2007b}. While it should not be surprising that a highly heterogenous processes could produce data indistinguishable from a homogeneous process on a different tree, the simplicity of these examples, and the limited heterogeneity they require, is perhaps more worrisome. If such examples were widespread, then there would be severe theoretical limits on our ability to detect when a heterogeneous process is acting. Moreover, heterogeneous processes on one tree might routinely mislead us into thinking data arose on a different tree. We have encountered researchers who, not surprisingly, find this possibility alarming.

In discussing mixture models, it is useful to distinguish between \emph{single-tree mixture models}, in which all sites evolve along the same topological tree but perhaps with different branch lengths, rate matrices, and base distributions, and \emph{multitree mixture models}, in which sites may evolve along different topological trees (as is appropriate when recombination, hybridization, or lateral gene transfer, occurs). Though
the explicit examples mentioned above are single-tree mixtures, mimicking by multitree mixtures is of course also a possibility.

In this work we investigate the possibility of mimicking, with the intent of understanding its origin and whether it should be a major concern. Because the question of whether mimicking occurs is closely related to the question of identifiability of parameters for mixture models, we begin with a review of the literature addressing the latter. 
Next we establish that a limited amount of heterogeneity in a single-tree mixture cannot mimic evolution on a different tree in most relevant circumstances.
We show how known examples of non-identifiability of trees due to mixture processes  arise from a readily understood issue of \emph{local over-parameterization}. 
Finally, for certain group-based models (Jukes-Cantor and Kimura $2$-parameter) we also obtain results indicating that if mimicking does occur for  multitree mixtures, then it is not entirely misleading. In the case of fully-resolved trees, any mimicking distribution can only agree with a distribution coming from one of the topological trees appearing in the mixture.

\subsection*{Mixture models and identifiability}

Model-based phylogenetic inference from sequence data requires compromises between simplicity and biological realism.  Typical current modeling assumptions include that all sites evolve on a single tree, according to the same substitution process, often with a simple $\Gamma$-distributed scaling of rates across the sites. While one can easily formulate models allowing  more complexity, the additional parameters this introduces can be problematic. Not only is software likely to require longer run-times, but one also risks
`overfitting' of finite data sets and thus interpreting stochastic variation as meaningful signal.

As larger data sets become more common, one might be less concerned with the threat of overfitting, and thus attracted to the use of more complex models. However, there are theoretical problems which can also prevent a complex model from being useful for inference, no matter how much data one has. If two or more distinct values of some parameter ---
the topological tree relating the taxa, for instance --- can lead to exactly the same expectations of data, then that parameter fails to be \emph{identifiable}. Without identifiability, even given access to unlimited data generated exactly according to the model, no method of inference will be able to dependably determine the true parameter value. In contrast, if a parameter is identifiable, then under very mild additional assumptions, the standard frameworks of maximum likelihood and Bayesian inference can be shown to be \emph{statistically consistent}. That is, assuming again that the model faithfully describes the data generation process, as the size of a data set is increased, the probability of these methods leading to an accurate estimate of the parameter approaches 1.

Of course the notion of statistical consistency says nothing about how statistical inference will behave when the process generating the data is not captured fully by the model chosen to analyze it (i.e., when the model used in the analysis is \emph{misspecified}). Nonetheless, consistency is generally viewed as a basic prerequisite for choosing an inference method, since without it a method is not sure to give good results even under idealized circumstances. As no tractable statistical model is likely to ever capture the full complexity of the processes behind sequence evolution, some model misspecification will always be with us. The inference task then depends on formulating models with enough complexity to capture the main processes we believe to be at work (thus minimizing misspecification), but which have identifiable parameters (so that in a more perfect world our inference methods  would not fail).

Unfortunately, it is not hard to conceive of data sets for which the modeling assumptions underlying today's routine analyses are strongly violated.  For instance, different parts of a single gene sequence might undergo rather different substitution processes, perhaps due to different substructures of the protein they encode. Alternatively,  lateral transfer of genetic material may have resulted in sequences that are amalgams of those evolving on different trees. Analyzing such data under a standard model simply assumes that neither of these has occurred, and so is an instance of misspecification. While one would hope there would be some indication of this as  the analysis is conducted --- perhaps by a poor likelihood score or poor convergence of a Bayesian MCMC run --- there is no guarantee that an obvious sign will appear.
 
An alternative is to consider \emph{mixture models}, which explicitly allow for such heterogeneity in the data. Mixtures consider several classes of sites which might each evolve according to a distinct process, either  on the same topological tree (a \emph{single-tree mixture model}), or on possibly different trees (a \emph{multitree mixture model}). In both cases the use of a mixture model differs from a partitioned analysis of data, in which the researcher imposes a partitioning of
the sites into classes, each of which must evolve according to a single standard model. For a mixture model, there is no \emph{a priori} partitioning; instead, the class to which a site belongs is treated as a random variable. The probability that any site is in a given class is then a parameter of the model, and thus to be inferred.

The single-tree GTR$+\Gamma$(+I) model is a familiar, but highly restricted type of mixture, with few parameters, that  is commonly used in data analysis. Only recently  \cite{Chai2011} completed a rigorous proof that the parameters of this model, including the tree topology, are identifiable from its probability distributions in most cases, and thus that it gives consistent inference under maximum likelihood.
However, the special case of the F81$+\Gamma$+I  submodel remains open \citep{AllmanAneRhodes07,Steel2009}.

On the other hand, a single-tree rate-variation model in which the rate distribution was allowed to be arbitrary was one of the earliest mixture models seen to be problematic, as every tree can produce the same distribution of site patterns \citep{Steel1994}.
The no-common-mechanism (NCM) model introduced by \cite{TufSte1997} provides another example of a mixture in which distributions
do not identify trees. However, these models are rather unusual, in that the number of their parameters grows with sequence length. This extreme over-parameterizaton is well understood, as is the implication that these models do not lead to statistically consistent inference under a maximum likelihood framework. (\cite{SteelSIN} offers a more complete and subtle discussion of NCM models and inference.) Of course these models were introduced to elucidate theoretical points, and were not intended  for data analysis. 
\medskip

Much recent work on mixture models has focused on those with a finite (though perhaps large) number of mixture components, allowing 
more heterogeneity among the classes than the simple scaling of the rate variation models.
Several papers have shown that inference from data generated by a mixture process can be poor if the analysis is based on a misspecified non-mixture model \citep{KT2004,MosVig,MosVig2}. The examples in these works indicate that we may be misled if we ignore the possibility of such heterogeneity.
This point is further underscored by \cite{Matsen2007}, who discuss why analysis with a misspecified non-mixture can lead to erroneous inference in some specific circumstances. As there is no general reason why one should expect good inference with a misspecified model, to our mind these works primarily indicate the importance of further study of mixture models, so that they may be applied intelligently when substantial heterogeneity is possibly present.

However, several works have indicated that models with a finite number of mixture classes may have theoretical shortcomings as well.  Working with no restriction on the number of classes, \cite{StefVig2007,StefVig2007b} emphasize that unless a model is special enough that there are  linear inequalities (which they call \emph{linear tests}) distinguishing between unmixed distributions arising on different trees, then there will be cases in which tree topologies cannot be identified from single-tree mixture distributions. 
\cite{Matsen2008} explore this more particularly for the Cavender-Farris-Neyman (CFN) 2-state symmetric model.

While there is no doubt that certain mixtures are problematic due to the failure of identifiability for some parameter choices, whether this is really of great practical concern is in fact not at all clear from the results mentioned so far. 
Thoughtful use of mixture models for data analysis has seemed to perform well for a number of research groups \citep{RonquistH2003,PagMea, PagelMeade,HuelSuch,LeLaGa,WLSR,EvansSullivan2012}. While publication bias against failed analyses could be responsible for a lack of reports of difficulties with mixture models in the literature, we also have not heard of such problems through our professional interactions. Of course this does not rule out the possibility that data is produced by even more heterogeneous processes that mimic those assumed in the analysis, and thus mislead us into believing an adequate model has been chosen.

Several papers \citep{Allman2006, APRS, RhodSull2011} have given a strong theoretical indication that problematic mixtures, for which trees are non-identifiable, are quite rare. Using algebraic techniques building on the idea of phylogenetic invariants, these works show in a variety of contexts that
mixture distributions cannot mimic distributions arising on other trees, for generic choices of numerical parameters. `Generic'
here has a precise meaning that informally can be expressed as ``if the model parameters are chosen at random, and thus do not have any special values or relationships among themselves.'' More formally, the set of exceptional parameters leading to non-identifiability is of strictly smaller dimension than the full parameter space. Thus if the true parameters were chosen by throwing a dart at the parameter space, with probability 1 they would lie off that exceptional set. \cite{RhodSull2011} give an upper bound on the number of classes that, for a quite general model, ensures generic identifiability of the trees in all single-tree and in many multitree mixtures. This bound is exponential in the number of taxa, and likely to be larger than the number of classes one would actually use in data analysis.

 While these positive theoretical results indicate one should seldom encounter problems with the judicious use of a mixture model in data analysis, one may still worry about the possible exceptions.
The exceptional cases are generally not explicitly characterized in these papers, and the arguments used to establish that they form a set of lower dimension are rather technical. The intuition of the authors is that the potential exceptional set one could extract from these works is likely to be much larger than the true exceptional set, 
as an artifact of the techniques of proof. 
Moreover, experience with other types of statistical models outside of phylogenetics (e.g., hidden Markov models, Bayesian networks) with similar exceptional sets of non-identifiability has shown they can still be quite useful, and are generally not problematic for data analysis.

\medskip
\subsection*{Mimicking and identifiability}

Considering models with a small number of mixture classes,
 \cite{StefVig2007,StefVig2007b} and \cite{Matsen2007} give explicit examples of parameter choices in certain 2-class CFN single-tree mixture models 
that lead to exactly the same unmixed probability distributions as a standard model on a  different tree. 
Since the unmixed model is a special case of a 2-class single-tree mixture (in which one class does not appear, due to a mixing parameter of 0, or alternatively in which the two classes behave identically), one interpretation of this result is a failure of tree identifiability
for 2-class CFN single-tree mixtures. Indeed, this example shows one cannot have identifiability across all of parameter space for this model, and thus that the generic identifiability mentioned in the last section is the best one can establish.

Another interpretation of the example, emphasized by the term `mimicking' used by \cite{Matsen2007}, is that  we could not distinguish data produced by the heterogeneous model from that produced by the unmixed one, and thus would have no indication that we should consider a mixture process as underlying the data. The simpler unmixed model would already fit data well, and we might not even consider the possibility of heterogeneity misleading us. (Of course performing an analysis of such data under the mixture model would not help us anyway, as the tree is not identifiable under it for the specific numerical parameters generating the data.)

Simpler models are nested within those allowing more heterogeneity and, as this example shows, the possibility of mimicking arises because identifiability may not hold for  all parameter values of the more complex model. The results of \citet{Allman2006, APRS, RhodSull2011}, which establish generic identifiability of  mixture models, therefore indicate that mimicking should be a rare phenomenon, requiring very special parameter choices in the more complex model. If a heterogeneous model has been shown to have generically identifiable parameters, then provided its parameters are chosen at random the probability of it mimicking a submodel is 0.
Nonetheless, if only generic identifiability of parameters of a mixture model is known, without an  explicit characterization of those special parameter choices leading to non-identifiablity, then we still have a less-than-solid understanding of when mimicking can occur.

In subsequent sections we give mathematical justification --- with no cryptic assumptions of genericity of parameters --- that a limited amount of heterogeneity in a single-tree mixture cannot mimic evolution on a different tree in most relevant circumstances.
We also show how examples of non-identifiability of trees due to mixture processes can arise from a readily understood issue of \emph{local over-parameterization}. This explains
the 2-class mimicking examples of  \cite{StefVig2007,StefVig2007b} and \cite{Matsen2007}, which are constructed for 2-state models whose parameter space is of larger dimension than the distribution space for a 4-taxon tree. However, this is not the setting in which most data analysis is likely to take place. For 4-state models encompassing those such as the general time-reversible (GTR) which are in common use, we show  even 3-class mixtures cannot mimic non-mixtures. While these positive identifiability results do not encompass the large number of mixture components allowed for generic parameters in the identifiability results of \cite{RhodSull2011}, by excluding the possibility of exceptions they are, in some sense, more complete.
Finally, for certain group-based models (Jukes-Cantor and Kimura $2$-parameter), for which linear tests exist, we also obtain results indicating that if  mimicking does occur for  multitree mixtures, then it is not entirely misleading. In the case of fully-resolved trees, any mimicking distribution can only agree with a distribution coming from one of the topological trees appearing in the mixture.
 
\medskip

The mathematical tools we use to obtain our results involve the polynomial equalities called phylognetic invariants, which have been extensively studied for both the group-based models and the general Markov model, and mixtures built from them. 
However, we
supplement these with some polynomial \emph{inequalities}. While the potential usefulness of inequalities was made clear even in the seminal paper of \cite{CF87} which introduced invariants, their study unfortunately remains much less developed than the study of invariants. Though a deeper understanding of inequalities for both unmixed and mixture models would be highly desirable, here we make do with a few \emph{ad hoc} ones.


\section*{Phylogenetic Mixture Models}

In this section, we describe the class of phylogenetic models that we study. Our definition of an unmixed phylogenetic model is broad,  encompassing most standard phylogenetic models such as the GTR, as well as those studied by \cite{StefVig2007,StefVig2007b}, \cite{Matsen2007}, and \cite{Matsen2008}. Informally, we consider continuous-time models, but do not require time-reversibility or stationarity, and allow the substitution process to change at a finite set of points on the tree. Such relaxations of the usual modeling assumptions have appeared in several works \citep{YangRoberts95,GG98,YapSpeed05}.

\smallskip

We assume that the random variables modeling characters have $\kappa\ge 2$ states, the most important values being $\kappa = 4$ (DNA models), $\kappa =2$ (purine/pyrimidine models), and $\kappa =20$ (protein models).

By a \emph{rate matrix} for  a state substitution process we mean a $\kappa \times \kappa$ matrix with nonnegative off-diagonal entries, whose row sums are all zero. (To fix a scaling, one may also impose some normalization convention.)
Such a rate matrix $Q=(q_{ij})$ generates a continuous-time $\kappa$-state Markov chain.   
Associated with $Q$ is  a directed graph, $G_{Q}$, on nodes $\{1,2, \ldots, \kappa\}$ representing states, which has an edge $i \rightarrow j$ if, and only if, $q_{ij} \neq 0$.  The process defined by $Q$  is said to be \emph{irreducible} if $G_{Q}$ is strongly connected, that is, there is a directed path from node $i$ to node $j$ for all $i, j $. Informally, this means it is possible to transition from any state to any other state, by possibly passing through other states along the way. Irreducibility guarantees that for all $t>0$ the discrete-time Markov transition matrix $\exp(Qt)$ has strictly positive entries.  Of course $\exp(Qt)$ is the identity matrix when $t=0$, and so has zero entries. 
\smallskip

Consider an unrooted, combinatorial, phylogenetic tree, $T$, in which we allow polytomies.
Then by the \emph{general continuous-time model} on $T$, we mean the following: First, possibly introduce a finite number of degree 2 nodes (in order to model a root, and points where the state substitution process changes) along any of the edges of $T$ to obtain $T'$. Then choose
some node to serve as a root of $T'$, and make
any assignment of a strictly positive $\kappa$-state distribution $\boldsymbol \pi$ at the root. Irreducible rate matrices $Q_i$ and edge lengths $t_i\in \mathbb R_{\geq 0}$ are assigned to each edge $i$ of $T'$. This notion is more general than is often used in most practical data analysis, since 1) $\boldsymbol \pi$ need not be the stationary distribution of any $Q_i$, and 2) the $Q_i$ may be different for each edge; we do not assume a common process across the tree. We at times restrict to considering only irreducible rate matrices  of a certain form (\emph{e.g.}, Jukes-Cantor, or GTR) and specialized $\boldsymbol \pi$, in order to draw conclusions about submodels.

If numerical model parameters are specified as above, then the Markov transition matrix on edge $i$ of $T'$ is $M_i=\exp(Q_it_i)$. If $T''$ denotes the tree obtained from $T'$ by suppressing non-root nodes of degree 2, and edges $i, i+1,\dots,i+r$ of $T'$ become a single edge of $T''$, then one defines a Markov matrix on that edge of $T''$ as the product $ M_iM_{i+1}\cdots M_{i+r}$.  From the assumption of irreducibility of rate matrices we immediately obtain the following.

\begin{lemma}\label{lem:positive}
Consider any choice of general continuous-time model parameters on a phylogenetic tree $T$. Then the Markov transition matrices associated to the edges of $T'$ and $T''$ are each either the identity matrix, or a nonsingular matrix with strictly positive entries.
\end{lemma}

The root distribution and  collection of edge transition matrices on $T''$ determine the probabilities of any site pattern occurring in sequence data. For instance, in a 5-taxon case of DNA sequences, the $4^5$ site patterns $AAAAA$, $AAAAG$, \dots, $TTTTT$ will be observed with probabilities that can be computed from the base frequencies (the entries of $\boldsymbol \pi$), and probabilities of various base substitutions over edges of the tree (the entries of the $M_i$).
The \emph{probability distribution} for a choice of general continuous-time model parameters for a fixed tree $T$ is then just the vector of the probabilities of all such site patterns. In the 5-taxon DNA case, for example, it is an ordered list of $4^5$ numbers describing expected frequencies of site patterns assuming the given parameter values.

By $\calm_{T}$ we denote the set of all probability distributions arising on $T$ for all choices of general continuous-time parameters. One can think of this object as encapsulating descriptions of all the infinite data sets that might be produced on the topological tree $T$, regardless of the specific base distribution, rate matrices, and edge lengths used. It is thus a basic theoretical object relating the general continuous-time substitution process on $T$ to data, without regard to specific numerical parameters. We therefore refer to $\calm_{T}$ as \emph{the general continuous-time model} on $T$. (Later in this paper, we use the same notation for a submodel obtained by restricting parameters to a specific form, such as Jukes-Cantor, but the distinction will be clear from the context.)

The  \emph{open phylogenetic model}, $\calm_{T}^{+} \subseteq \calm_{T}$, is the subset of distributions obtained by requiring that no internal branch lengths are
zero, that is all $t_{i}>0$ except possibly for pendant edges.   Since we allow trees to have polytomies, any distribution in  $\calm_{T}$  is contained in the open model for a possibly different tree; one merely contracts all internal edges of $T$ which were assigned branch length zero, thus introducing new polytomies. 

\medskip

If $\calt = \{T_{1}, \ldots, T_{r} \}$ is a multiset of topological trees, then the \emph{mixture model} on $\calt$ is the set $\calm_{\calt}$ of all probability distributions of site patterns of the form 
$$s_1p_1+s_2p_2+\cdots +s_rp_r,$$
where $p_i\in \calm_{T_i}$ is a probability distribution arising on $T_i$, and the $s_i\ge 0$ are \emph{mixing parameters} with $s_1+s_2+\cdots +s_r=1$. The $s_i$ can be interpreted as the probabilities that any given site is in class $i$, while $p_i$ is the vector of site pattern probabilities for that particular class.
The \emph{open mixture model}  $\calm_{\calt}^{+}$
is defined similarly, with $p_i\in \calm_{T_i}^+$.
Note that in the open mixture model we allow all mixing parameters, so that some mixture components may in fact not appear if an $s_i=0$.  If  all mixing parameters are required to be strictly positive, we denote the set of distributions by $\calm_{\calt}^{++}$.

\section*{Results}
\subsection*{Single-tree Mixture Models}

\cite{Matsen2007} and \cite{StefVig2007b} showed that under the CFN model it is possible for a $2$-class mixture on a single topological tree (that is, $\calt = \{T,T \}$) to produce distributions matching those of an unmixed model on a different tree.  \cite{Matsen2008} showed that this is possible if, and only if, the trees involved differ by a single NNI move.

Our main result in this setting shows that these possibilities are essentially a ``fluke of low dimensions,'' tied to the 2-state nature of the CFN model. Models with larger state spaces, such as the 4-states of DNA models, cannot exhibit such mimicking behavior with such a small number of mixture components.  In a subsequent section a further analysis will show that this CFN mimicking is a consequence of local over-parameterization.  

\begin{thm}\label{thm:nonmimic}
Consider the $\kappa$-state general continuous-time phylogenetic model.
Let $\calt$ consist of $\kappa -1$ copies of tree $T_{1}$, and $\cals$ consist of a single tree $T_2$.  Then $\calm_{\calt}$ and $\calm_{\cals}^{+}$ have no distributions in common, and thus mimicking cannot occur, unless $T_{1}$ is a refinement of $T_{2}$.
\end{thm}

Note that while the mixture on $\calt$ in this theorem has all classes evolving on the same topological tree, no further commonality across classes is assumed. The individual classes may not only have different edge-lengths associated to the tree, but also different base distributions and rate matrices.

A closely related identifiability result was already known to hold for generic choices of parameters in a slightly broader setting \citep{Allman2006}, so the contribution here is to remove the generic assumption. 
Note that for the important case of $\kappa =4$, corresponding to DNA models, this implies that we cannot have a 2- or 3-class mixture mimic the distribution on a single tree unless we allow zero length branches in the mixture components. This indicates the examples of \cite{Matsen2007} and \cite{StefVig2007,StefVig2007b} cannot be generalized to 4-state models, without passing to at least a 4-class mixture.


\subsection*{Local Over-parameterization} \label{sec:local}

Note that the examples of \cite{Matsen2007} and \cite{StefVig2007b} are allowed by Theorem \ref{thm:nonmimic}, since they are constructed for a model with $\kappa=2$ and $\calt$ a 2-element multiset. To see why the existence of such examples should not be too surprising, it is helpful to first consider an unrooted 4-leaf tree $T$ and perform a parameter count for the CFN model. A 2-class single-tree mixture on $T$ can be specified by 11 numerical parameters: for each class there are 5 Markov transition matrices with 1 free parameter (the edge length) each, and 1 additional mixing parameter. However any 4-taxon CFN mixture distribution on any 4-taxon tree lies in a certain 7-dimensional space, due to the symmetry of the model. An 11-dimensional parameter space is thus collapsed down to a subset of a 7-dimensional distribution space. Although this does not prove every distribution with such symmetry must arise from this 2-class mixture, the excess of parameters suggests that it is likely that many do. As a result, one suspects at least some non-mixture distributions on trees different from $T$ are likely to be mimicked by this 2-class mixture. This suspicion is then confirmed by explicit examples.

When a tree has many more leaves, however,  a similar parameter count for the 2-class CFN mixture can fail to indicate potential problems, since the number of model parameters grows linearly with the number of leaves, while the number of possible site patterns grows exponentially.
However, we show below that one can extend mimicking examples on small trees to larger trees, thus creating what might at first appear to be more unexpected instances of mimicking. We refer to such examples, where mimicking is produced first on a small tree by allowing
an excessive number of mixture components, and then extended to larger trees,
as arising from \emph{local over-parameterization}. This notion can be used to produce many new examples of the mimicking phenomenon, on single- or multitree mixtures.  

We distinguish here between three types of mimicking, of different degrees of severity.  For notational convenience we use $\calm_{\calt}^{*}$ to denote any of the models $\calm_{\calt}, \calm_{\calt}^{+},$ or $\calm_{\calt}^{++}$.

\begin{defn}
A mixture model $\calm_{\calt}^{*}$  \emph{weakly mimics} distributions in $\calm_{\cals}^{*}$ if $\calm_{\calt}^{*}$ and $\calm_{\cals}^{*}$
have no distributions in common, \emph{i.e.}, if
$\calm_{\calt}^{*} \cap  \calm_{\cals}^{*} \neq \emptyset$.  A mixture model  $\calm_{\calt}^{*}$  \emph{strongly mimics} distributions in $\calm_{\cals}^{*}$ if $\dim \calm_{\calt}^{*} \cap  \calm_{\cals}^{*} = \dim \calm_{\cals}^{*}$.  A mixture model  $\calm_{\calt}^{*}$ \emph{completely mimics} distributions in $\calm_{\cals}^{*}$ if $  \calm_{\cals}^{*} \subseteq \calm_{\calt}^{*}$.
\end{defn}

Thus weak mimicking requires only a single instance of probability distributions arising on $\cals$ and $\calt$ matching, for a single pair of parameter choices for the models. Strong mimicking requires a neighborhood of distributions arising on $\cals$ to be matched by ones arising on $\calt$, so that all parameter choices near a specific pair lead to mimicking.
Complete mimicking
requires every distribution arising on $\cals$ to be matched by one arising on $\calt$, so that mimicking occurs for all parameter choices.

More informally, weak mimicking that is not strong can be viewed as unlikely to be problematic in practice, since it does not occur over a range of parameter values. Similarly, strong mimicking that is not complete may be a serious problem on parts of parameter space, but
is limited in not affecting all choices of parameters. Complete mimicking, however, means it is impossible to determine if any data fit by the mimicked model actually arose from the mimicking one.

\medskip

To make the idea of local over-parameterization precise, we need the concept of a \emph{fusion tree}, as depicted in Figure \ref{fig:fusion}. Informally, one considers a `core' tree with only a few leaves, and then enlarges it to relate many more taxa, by attaching rooted trees to the leaves.
Let $T$ be the core tree relating taxa $X$.  For each $x \in X$, let $B_{x}$ be a rooted tree with taxon set $A_{x}$, where the $A_x$ have no elements in common.   A set of such trees $\calb = \{B_{x}: x \in X\}$ is called a set of \emph{fusion ends} for $X$.  The \emph{fusion tree} $T^{\calb}$, with leaf set $\cup_{x \in X} A_{x}$, is obtained from $T$ and $\calb$ by identifying each leaf  $x$ of $T$ with the root of $B_{x}$. In short, the fusion tree $T^{\calb}$ is obtained by fusing the trees
in $\calb$ onto the leaves of $T$.

\begin{figure}[h]
\begin{center}
\includegraphics[height=3.in]{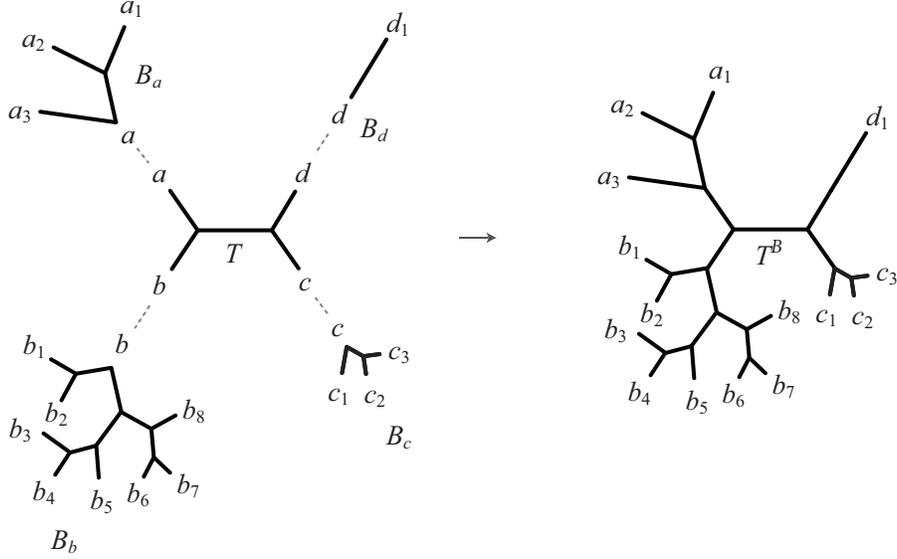} 
\end{center}
\caption{The fusion tree $T^\calb$ is constructed from a tree $T$  with leaf set $X=\{a,b,c,d\}$ and a set $\calb=\{B_a,B_b,B_c,B_d\}$ of fusion ends for $X$. The construction using $\calb$ could be applied to any of the quartet trees with leaf set $X$, yielding fusion trees differing by an NNI move from the $T^\calb$ shown here. This process underlies the extension of mimicking examples on small trees to larger ones.}\label{fig:fusion} 
\end{figure}

If $\calt$ is a collection of trees with the same leaf set $X$, and $\calb =  \{B_{x}: x \in X\}$ is  a set of fusion ends for $X$, let $\calt^{\calb}$ be the multiset $\calt^{\calb}  = \{T^{\calb} : T \in \calt \}$ of fusion trees. Thus all trees in $\calt^{\calb}$ display the same topological structure for the subtrees of the fusion ends, but can differ in their cores. The following propositions allow us to pass mimicking properties from small trees to large trees.

\begin{prop}\label{prop:fusion1}
Suppose for a taxon set $X$ that $\calm^{*}_{\calt}$ weakly mimics $\calm^{*}_{\cals}$, and that $\calb$ is a set of fusion ends.  Then $\calm^{*}_{\calt^{\calb}}$ weakly mimics $\calm^{*}_{\cals^{\calb}}$.
\end{prop}

\begin{proof}
Any distribution $q\in\calm_{\calt}^{*} \cap  \calm_{\cals}^{*}$ arises from parameters on the trees in $\calt$, as well as from parameters on the trees in $\cals$.  Retain these parameters on the corresponding edges of the trees in $\calt^{\calb}$ and $\cals^{\calb}$. Choose  a length and rate matrix for each edge of each tree
in $\calb$, thus determining probabilities of site patterns at the leaves of the fusion ends conditioned on root states. Use these choices for the corresponding edges in the individual fusion subtrees in $\calt^{\calb}$ and $\cals^{\calb}$. With the same mixing parameters as led to $q$, these parameters give rise to a distribution $ q^\calb \in
\calm^{*}_{\calt^{\calb}}\cap\calm^{*}_{\cals^{\calb}}$.
\end{proof}

Under an additional assumption that the mimicked model $\calm_\cals^*$ is unmixed, more can be said.

\begin{prop}\label{prop:fusion2}
Let $\cals = \{T \}$, be a single tree, and suppose that $\calm^{*}_{\calt}$ strongly mimics (or completely mimics) $\calm^{*}_{\cals}$. Then for any set $\calb$ of fusion ends for $X$,  $\calm^{*}_{\calt^{\calb}}$ strongly mimics (or completely mimics) $\calm^{*}_{\cals^{\calb}}$.
\end{prop}

\begin{proof}
This follows from the same argument as was given for Proposition \ref{prop:fusion1}, with the additional observation that the parameters assigned to
edges in the fusion ends can be varied arbitrarily.  Since $\cals$
consists of a single tree, this will give a give a full dimensional set of distributions
in $\calm^{*}_{\cals^{\calb}}$ which are mimicked by distributions in 
$\calm^{*}_{\calt^{\calb}}$.  If $\calm^{*}_{\calt}$ completely mimics $\calm^{*}_{\cals}$,
note that every distribution in $\calm^{*}_{\cals^{\calb}}$ arises from our construction
so that  $\calm^{*}_{\calt^{\calb}}$  completely mimics $\calm^{*}_{\cals^{\calb}}$. 
\end{proof}

These propositions allow the construction of explicit examples of mimicking behavior on large trees from those found on small trees.  A typical result of this type, using quartet trees as the core, is:

\begin{thm}\label{thm:nni}
Let $\calt$ consist of $r$ copies of the quartet tree $T_{12|34}$, and $\cals$ consist of $s$ copies of the quartet tree $T_{13|24}$.  Let $T$ and $T'$ be trees with at least 4 leaves that differ by an NNI move,  $\calt'$ consist of $r$ copies of $T$, and $\cals' $ consist of $s$ copies of $T'$. 

If  $\calm^{*}_{\calt}$ weakly mimics $\calm^{*}_{\cals}$, then   $\calm^{*}_{\calt'}$ weakly mimics $\calm^{*}_{\cals'}$.
Furthermore, if  $\calm^{*}_{\calt}$ strongly (or completely) mimics $\calm^{*}_{\cals}$ and $s = 1$, then $\calm^{*}_{\calt'}$ strongly (or completely)  mimics $\calm^{*}_{\cals'}$
\end{thm}

\begin{proof}
Two trees differ by an NNI move if, and only if, they are obtained from applying fusions to two differing quartet trees.  Hence, we can apply Propositions \ref{prop:fusion1} and \ref{prop:fusion2}.
\end{proof}

In particular, Theorem \ref{thm:nni} implies that if quartets give mimicking behavior, then we will have mimicking behavior on trees of arbitrary size.  (Note conversely that Theorem 31 of \cite{Matsen2008} shows that the only way 2-class single-tree CFN mixtures can mimic CFN non-mixtures on large trees is through such a process applied to quartet over-parameterization.)

Consider now the general continuous-time model on a 4-leaf tree. With 5 edges, a distribution is specified by  $\approx 5 \kappa^{2}$ numerical parameters.  Since there are no linear tests for this model, and the probability distribution lies in a space  of dimension $\kappa^4-1$, we expect that a mixture of more than $\approx \kappa^{2}/5$ components will include an open subset of the probability simplex. Hence such a model is likely to display mimicking behavior.  Thus some sort of mimicking seems unavoidable for even moderately sized mixtures. To illustrate,  with DNA sequences and $\kappa=4$, an unmixed model  is specified by 63 parameters, so the 4-class
mixture model has enough parameters ($4\times 63+3=255$) that it is likely to include a full-dimensional subset (since $\kappa^4-1=255$) and produce mimicking.

Note that mimicking of the sort produced by local over-parameterization need not be limited to that arising from quartet trees as in the specific example above. With enough mixture components, for some models it may be possible for a mixture on a relatively small tree to mimic a distribution from another tree, differing by more than a single NNI move. This mimicking would again extend to larger trees, using the fusion process of Propositions \ref{prop:fusion1} and \ref{prop:fusion2}. 

From a practical perspective, however, mimicking through local over-parameterization seems unlikely to be much of an issue in most data analyses, since the mixture parameters leading to it require that the mixed processes differ only on a  small part of the tree, and are identical elsewhere. Researchers studying biological situations in which this might be plausible should, however, be aware of the possibility.

Finally, we emphasize that we have not shown that local over-parameterization is the only possible source of mimicking. It would be quite interesting to have examples of mimicking of other sorts, or extensions of Theorem 31 of \cite{Matsen2008} to other models and more mixture components.

\subsection*{Models with Linear Tests}

An early motivation for the study of linear invariants for phylogenetic models was that they are also invariants for mixture models on a single tree, and thus offered hope for determining tree topologies even under heterogeneous processes across sites. 
While poor practical performance \citep{Huel95} even in the unmixed case led to their abandonment as an inference tool, they remain useful for theoretical purposes. However, among the commonly-studied phylogenetic models, the Jukes-Cantor (JC) and Kimura $2$-parameter (K2P) models are the only ones which possess phylogenetically-informative linear invariants. 

\cite{StefVig2007} used these linear invariants and the observation that they can be used to give linear tests, to show that if $\cals$ and $\calt$ are multisets each consisting of a single repeated $n$-leaf binary (fully-resolved) tree, and these trees are different, then $\calm_{\calt}^{+}$ and $\calm_{\cals}^{+}$ have no distributions in common, regardless of the number of mixture components.  We next explore the extent to which these results can be extended to nonidentical tree mixtures for the JC and K2P models.

\begin{thm}\label{thm:kimura}
Consider the Jukes-Cantor and Kimura $2$-parameter models.
Let $\cals$ be a multiset of many copies of tree $T_1$ on $X$, and $\calt$ an arbitrary multiset of trees on $X$.  If $\calm_{\cals}$ and $\calm_{\calt}^{++}$ contain a common distribution, then for every four element subset $K \subseteq X$, and for all $T \in \calt$, either $T|_{K}$ is an unresolved (star) tree, or $T|_{K} = T_{1}|_{K}$. Thus all trees in $\calt$ have $T_1$ as a binary resolution.  

Furthermore, if  all trees $T \in \calt$ are binary, and $T_{1} \notin \calt$ then $\calm_{\cals}$ and $\calm_{\calt}^{+}$ have no distributions in common.
\end{thm}

Informally,  the last statement of this theorem states that arbitrary multitree phylogenetic mixtures on fully-resolved trees cannot mimic mixtures on a single tree, unless that tree appears in some component of the mixture. Thus if one erroneously assumed such a mimicking distribution was from a single-tree mixture, the single tree one would recover would in fact reflect the truth for at least one mixture component.

In the case that $\cals=\{T_1\}$, so $\calm_\cals$ is not a mixture but rather a standard model, for the JC and K2P models
this again rules out any mimicking examples of
the sort \cite{Matsen2007} and \cite{StefVig2007b} give for CFN, unless one allows zero length branches.  This clearly indicates the special nature that any such exceptional cases must have.

Our final theorem shows that the special case of mimicking allowed by Theorem \ref{thm:kimura} actually occurs for the Jukes-Cantor model. We provide a construction of such mimicking, where $\mathcal{T}$ contains nonbinary trees that are degenerations of the tree
$T$.

\begin{thm}\label{thm:contract}
Let $T$ be a tree with internal vertex $v$ which is adjacent to three other vertices
$u_{1}, u_{2}, u_{3}$.  For $i = 1,2$, let $T_{i}$ be the 
tree obtained from
$T$ by contracting the edge $u_{i}v$.  Let $\cals = \{T\}$ and $\calt = \{T_{1}, T_{2}\}$.
Then, under the Jukes-Cantor model $\calm^{*}_{\calt}$ completely
mimics $\calm^{*}_{\cals}$.
\end{thm}

\section*{Conclusion}

Interest in the analysis of data sets produced by heterogeneous evolutionary processes is likely to grow, as larger data sets are more routinely assembled. With the increased complexity of heterogeneous models, however, comes the potential loss of ability to validly infer even the tree (or trees) on which the models assume evolution occurs. Such a failure can happen not because data is insufficient to infer parameters well, but rather due to theoretical shortcomings such as non-identifiability of parameters or mimicking behavior. Extreme instances of mixture models, such as the ``no common mechanism'' model, are known to exhibit such flaws.

While one might wish that software allowing the use of mixture models could warn one if a chosen model is theoretically problematic, this is of course asking too much. A programmed algorithm applied to a non-identifiable model still runs, and produces some output. Programming decisions that have no effect on the output when an identifiable model is used may result in certain biases under a non-identifiable one, so that, under a maximum likelihood analysis for instance, it appears that a particular parameter value has been inferred even though other values produce the same likelihood. In the same vein, a Bayesian MCMC analysis may have poor convergence, and the posterior distribution may be highly sensitive to the choice of prior. Thus theoretical understanding of identifiability issues are essential.

Establishing which phylogenetic mixture models have few, or no,  theoretical shortcomings has proven difficult, but a collection of results
has now emerged that can at least guide a practitioner.  \citet{RhodSull2011} provide the largest currently-known bound on
how many mixture components can be used in a model before identifiability may fail, a bound that is exponential in the number of taxa.
However, this bound is established only for generic choices of parameters. While similar generic results for complex statistical models outside of phylogenetics
are generally accepted as indications a model may be useful, it is still desirable to understand the nature of possible exceptions.

Explicit examples have shown exceptions do indeed exist for phylogenetic mixtures, and in particular that the mimicking of an unmixed model  by a mixture can occur, even with fairly limited heterogeneity. However the structure of known examples is quite special, depending on what we have called local over-parameterization. We have also shown here that local over-parameterization provides a general means by which examples of lack of identifiability or mimicking can be constructed in the phylogenetic setting. While a simple check that the number of parameters of a complex model exceeds  the number of possible site patterns can serve as a indication of a failure of identifiability in other circumstances, this check may not uncover problems due to local over-parameterization.  

Although we do not believe problems due to mimicking through local over-parameterization are at all common in data analysis, those analyzing data which could plausibly be produced by heterogeneous processes should be aware of the possibility.  Mimicking due to local over-parameterization arises because of excessive heterogeneity of a mixture on a small core part of the tree, combined with homogeneity elsewhere in the tree. The plausibility of this occurring must be judged on biological grounds.
If a mixture remains heterogeneous over the entire tree, then by our understanding of generic identifiability of model parameters, mimicking should not occur, with probability 1.

Under more assumptions than those of \citet{RhodSull2011}, we have shown here that it is possible to rule out some undesirable model behavior. If the number of mixture components is small (3 or fewer for DNA models) then there can be no mimicking of an unmixed model by a single-tree mixture of general continuous-time models. Attempting to raise this bound would likely require carefully cataloging exceptional cases, including those arising from local over-parameterization and other causes (if they exist). The technical challenges of doing this may mean that theorems indicating exact circumstances under which a given mixture model may lack parameter identifiability will elude us for some time.

Finally, in the even more specialized setting of certain group-based models, previous work had shown that mixtures on one tree topology could not mimic those on another, even if arbitrarily many mixture components are allowed. We extended this in Theorem \ref{thm:kimura} to show that
a mixture on many different trees could not mimic that on a single tree unless there are strong relationships between the tree topologies.
Though investigations with these models have little direct applicability to current practice in data analysis, the insights gained provide some indications of how more complicated models might behave.

\section*{Funding}

The work of Elizabeth Allman and John Rhodes is
supported by the U.S. National Science Foundation (DMS
0714830), and that of Seth Sullivant by the David and Lucille
Packard Foundation and the U.S. National Science Foundation (DMS 0954865).

\section*{Acknowledgements}

This work was begun at the Institut Mittag-Leffler, 
during its Spring 2011 program  `Algebraic Geometry with a View Towards Applications.'
The authors thank the Institute and program organizers for both support and hospitality.
 

\bibliography{Mixtures}

 
\section*{Appendix: Mathematical Arguments}

To prove Theorem \ref{thm:nonmimic} we first handle the special case of $4$-leaf trees. We need the following definition.

\begin{defn} If $P$ is a probability distribution for a $\kappa$-state phylogenetic model on a $n$-taxon tree, we view it as 
an $n$-dimensional $\kappa\times \kappa\times \cdots \times \kappa$ tensor, or array, of probabilities, $P=(p_{i_1i_2\dots i_n})$, where the index $i_l$ refers to the state at leaf $l$. Then given any bipartition of the leaves into non-empty subsets $\{1,2,\dots,n\}=A\sqcup B$, the $A|B$ \emph{flattening} of $P$ is the $\kappa^{|A|}\times \kappa^{|B|}$ matrix $\flat_{A|B}$ with the same entries as $P$ but with rows indexed by state assignments to leaves in $A$, and columns indexed by state assignments to leaves in $B$.
\end{defn}

\begin{lemma}\label{lem:quartet}Consider 4-leaf trees $T_1$ with
split $12|34$, and $T_{2}$ either the tree with split $13|24$ or the star tree. Then the statement of Theorem \ref{thm:nonmimic} holds.
That is,
 $\calm_{\calt} \cap \calm_{\cals}^{+} = \emptyset$ unless $T_{2}$ is the star tree.

\end{lemma}

\begin{proof}
Let $p$ denote a probability distribution $p \in \calm_{\calt}$, which we consider as a $4$-dimensional tensor.  Consider the $\{1,2\}|\{3,4\}$ flattening $\flat_{12|34}(p)$, which is a $\kappa^{2} \times \kappa^{2}$ matrix.  From \cite{Allman2006} or \cite{Eriksson2005} it is known that if $p \in \calm_{\calt}$ then the rank of $\flat_{12|34}(p)$ is at most $\kappa(\kappa - 1)$.  

On the other hand, if  $T_2=13|24$ and $q \in \calm_{\cals}^{+}= \calm_{T_2}^{+}$, then the matrix $\flat_{12|34}(q)$ has a factorization as
\begin{equation}\flat_{12|34}(q) = (M_1 \otimes M_2)  {\diag}(N)  (M_3 \otimes M_4)\label{eq:flat}
\end{equation}
where $M_i, 1\le i\le 4$ are the transition matrices associated with the leaf edges in the tree, and $N=\diag(\boldsymbol \pi)M_5$  where $M_5$ is the transition matrix associated to the internal edge and we have assumed the tree root is at one end of that edge. Here $\diag(N)$ denotes a $\kappa^2\times\kappa^2$ diagonal matrix constructed with the entries of $N$ on its diagonal in an appropriate order. 
By Lemma \ref{lem:positive}, all transition matrices for the model $\calm_{T_2}$ are nonsingular.  Thus the $\kappa^{2} \times \kappa^{2}$ matrices $M_1 \otimes M_2$ and $M_3 \otimes M_4$ are nonsingular.  Also by  Lemma \ref{lem:positive}, for the open model $\calm_{T_{2}}^{+}$, the matrix ${\diag}(N)$ is nonsingular since all the entries of $\boldsymbol \pi$  and $M_5$ are nonzero.  Thus
if $q\in \calm_{T_{2}}^{+}$, $\flat_{12|34}(q)$ has rank $\kappa^2$. 

If $T_2$ is the star tree, then formula \eqref{eq:flat} still holds if one sets $M_5=I$. In this case the matrix ${\rm diag}(N)$ is singular, and the rank of $\flat_{12|34}(q)$ is $\kappa$.

These conditions on the rank of $\flat_{12|34}(q)$ now imply the desired conclusion.
\end{proof}

\begin{proof}[Proof of Theorem \ref{thm:nonmimic}]
If $T_{1}$ is a refinement of $T_{2}$, then one checks that $\calm_{\cals}^{+} \subset \calm_{\calt}$, by choosing the mixing weights as a standard unit vector, and setting edge lengths equal to zero on the edges appearing in $T_{1}$ but not $T_{2}$.  

So assume that  $T_{1}$ is not a refinement of $T_{2}$, yet $\calm_{\calt} \cap \calm_{\cals}^{+}$ is non-empty.  We may also assume that $T_{1}$ is a binary tree, by passing to a refinement,  as this only enlarges the mixture model.    
There exists a subset $K$ of four taxa such that the induced quartet trees $T_{1}|_{K}$ and $T_{2}|_{K}$ are different.  Marginalizing to $K$, since   $(\calm_{\calt})|_{K} =  \calm_{(\calt|_{K})}$ and $(\calm_{\cals})|_{K} =  \calm_{(\cals|_{K})}$, we have that $\calm_{(\calt|_{K})} \cap \calm_{(\cals|_{K})}^{+}$ is non-empty.

Now, by Lemma \ref{lem:positive} the transition matrices that arise in the resulting quartet trees will be products of nonsingular matrices
that either are the identity, or have all positive entries.  Thus each quartet tree transition matrix is nonsingular and can have zero entries if, and only if, it is the product of identity matrices. We now apply Lemma \ref{lem:quartet} to deduce that all the edge lengths along the internal edge of $T_2|_{K}$ must be zero.  But this contradicts the fact that we were working with the open model $\calm^{+}_{\cals}$.
\end{proof}

\medskip

To prove Theorem \ref{thm:kimura}, we recall a number of results about the JC and K2P models, including their descriptions in Fourier coordinates, and properties of linear invariants/tests for these models.

The JC, K2P (and K3P) models are group-based models, with a special structure governed by the finite abelian group $G=\zz_{2} \times \zz_{2}$.  We associate  nucleotides
with elements of this group via
$$
A = (0,0), \ C = (0,1),\  G = (1,0),\  T = (1,1).
$$
The discrete Fourier transform (also called 
Hadamard conjugation in this context) \citep{Hendy1989,Evans1993} is an invertible linear transformation that simplifies the parameterization of a group-based model.  In Fourier coordinates, $q_{g_{1} \ldots g_{n}}$, the parameterization is described as follows:  To each of the tree $T$'s splits  $A|B$
we associate a collection of parameters $a_{g}^{A|B}$ where $g \in G$.  Then
\begin{equation}\label{eq:Fparam}
q_{g_{1} \ldots g_{n}} =  \left\{
\begin{array}{cl}
\prod_{A|B } a^{A|B}_{\sum_{i \in A} g_{i}} & \mbox{ if }
\sum g_{i} =0 \\
0 & \mbox{otherwise.}
\end{array} \right. 
\end{equation}

\begin{prop}\label{prop:matsenineq}
Suppose that a transition matrix has the form $exp(Qt)$ where $Q$ is  a rate matrix for a $\zz_{2} \times \zz_{2}$ group-based model, $t > 0$, and $Q$ defines an irreducible Markov chain.  Then the Fourier parameters satisfy the constraints:
\begin{align*}
a^{A|B}_{A} &=1,\\
a^{A|B}_{C} & \geq  a^{A|B}_{G} a^{A|B}_{T},\\
a^{A|B}_{G} & \geq  a^{A|B}_{C} a^{A|B}_{T},\\
a^{A|B}_{T} & \geq  a^{A|B}_{C} a^{A|B}_{G},
\end{align*}
with $a^{A|B}_{C}, a^{A|B}_{G}, a^{A|B}_{T} \in (0,1)$. When $t = 0$, all parameters equal $1$.

Additionally,  under the K2P model, $a^{A|B}_{G} = a^{A|B}_{T}$, and under the JC model $a^{A|B}_{C} = a^{A|B}_{G} = a^{A|B}_{T}$.
\end{prop}

\begin{proof}
Let $Q$ be a rate matrix of K3P format and $H$ the associated $4 \times 4$ Hadamard
matrix, that is, for some $\alpha,\beta,\gamma \ge 0$, $\delta=-\alpha-\beta-\gamma$,
$$
Q = \begin{bmatrix}
\delta& \alpha & \beta & \gamma \\
\alpha & \delta & \gamma &  \beta  \\
 \beta & \gamma  & \delta & \alpha \\
 \gamma &  \beta  &\alpha & \delta  \\
\end{bmatrix} 
\quad \mbox{ and } \quad 
H  = 
\begin{bmatrix}
1 & \phantom{-}1 & \phantom{-}1 & \phantom{-}1 \\
1 & -1 & \phantom{-}1 & -1 \\
1 & \phantom{-}1 & -1 & -1 \\
1 & -1 & -1 & \phantom{-}1
\end{bmatrix}. 
$$

\phantom{-}
The Fourier coordinates for this model consist of the eigenvalues of the
matrix $\exp(Qt)$.  The matrix $H$ consists of the eigenvectors of the matrix
$Q$, and hence of $\exp(Qt)$.  We compute that $H^{-1}QH$ is the diagonal matrix
$\diag(0, -2 (\alpha +\gamma), -2(\beta +\gamma), -2(\alpha + \beta))$.
From this we deduce that the Fourier coordinates for this model are then
$$
a^{A|B}_{A} = 1,\   a^{A|B}_{C} = \exp(-2 t (\alpha +\gamma)),\  
a^{A|B}_{G} = \exp(-2 t (\beta +\gamma)),\  a^{A|B}_{T} = \exp(-2 t (\alpha +\beta)).
$$
Since $Q$ gives an irreducible Markov chain, at most one of $\alpha, \beta,$
and $\gamma$ can be zero, which implies that all of $a^{A|B}_{C}, a^{A|B}_{G}, a^{A|B}_{T} < 1$
when $t > 0$.
Furthermore, we see that the claimed inequalities hold, \emph{e.g.},
$$
a^{A|B}_{C} = \exp(-2 t (\alpha +\gamma)) \geq \exp(-2 t (\alpha + 2\beta +\gamma)) 
= a^{A|B}_{G}a^{A|B}_{T}.
$$ 
Note also that the K2P model consists of all rate matrices where
 $\alpha = \gamma$, which implies that $a^{A|B}_{G} = a^{A|B}_{T}$,
 and the JC models consists of all rate matrices where $\alpha = \beta = \gamma$,
 which implies that $a^{A|B}_{C} =a^{A|B}_{G} = a^{A|B}_{T}$.
\end{proof}

\begin{prop}\label{prop:quartets}
Let $T_{1} = T_{12|34}$, $T_{2} = T_{13|24}$, and $T_{3} = T_{14|23}$. 
Then under the JC and K2P models,  the polynomial
$$
l(q) = q_{GGGG} - q_{GGTT}
$$
satisfies the following properties:
\begin{enumerate}
\item $l(q) = 0$ for all $q \in \calm_{T_{1}}$,
\item $l(q) \geq 0$ for all $q \in \calm_{T_{i}}$, $ i = 2,3$,
\item $l(q) > 0$ for all $q \in \calm_{T_{i}}^{+}$, $i = 2,3$, and 
\item if $q \in \calm_{T_{i}}$, for $i=2$ or $3$, and $l(q) = 0$, then the branch length of the internal edge is zero.
\end{enumerate}
\end{prop}

\begin{proof}
To evaluate the polynomial $l(q)$, we substitute for $q$ the parametric expressions given in equation \eqref{eq:Fparam}.  
Denoting parameters for trivial splits $\{i\}|(\{1,2,3,4\}\smallsetminus \{i\})$ by  $a^i_g$, for  $q\in \calm_{T_1}$ we have
$$
l(q) =  q_{GGGG} - q_{GGTT} = a^{1}_{G}a^{2}_{G}a^{3}_{G}a^{4}_{G}a^{12|34}_{A} - 
a^{1}_{G}a^{2}_{G}a^{3}_{T}a^{4}_{T}a^{12|34}_{A}.
$$
Since $a^{A|B}_{G} = a^{A|B}_{T}$ in the JC and K2P models, the first claim follows.

If $q \in \calm_{T_{2}}$, to establish the remaining claims note
$$
l(q) =  q_{GGGG} - q_{GGTT} = a^{1}_{G}a^{2}_{G}a^{3}_{G}a^{4}_{G}a^{13|24}_{A} - 
a^{1}_{G}a^{2}_{G}a^{3}_{T}a^{4}_{T}a^{13|24}_{C}.
$$
Since $a^{A|B}_{G} = a^{A|B}_{T}$for the JC and K2P models, and $a^{13|24}_{A} =1$,  this expression factors as
$$
l(q) =a^{1}_{G}a^{2}_{G}a^{3}_{G}a^{4}_{G}(1 - a^{13|24}_{C}).
$$
By Proposition \ref{prop:matsenineq} all $a^{A|B}_{g} \in (0,1]$ , so $l(q)\geq 0$.  Moreover,
if all branch lengths are strictly positive, so is $l(q)$.
On the other hand, the only way this expression can  equal  zero with 
$q \in \calm_{T_2}$ is if $a^{13|24}_{C} = 1$.  But then Proposition
\ref{prop:matsenineq} implies  the length of the internal branch is zero.  

Similar arguments show the claims for $T_{3}$. 
\end{proof}

\begin{proof}[Proof of Theorem \ref{thm:kimura}]
Let $K$ be any four element subset of the taxa.  If 
$\calm_{\calt} \cap \calm_{\cals}^{+} \neq \emptyset$,
then when we marginalize to mixture models on the leaf set $K$ the corresponding intersection is also non-empty.  Since the claims of the theorem concern quartets, it suffices to restrict attention to the case of $n=4$ taxa.  

First suppose that the tree $T_{1}$ is fully-resolved.  By symmetry we may assume it is $T_{12|34}$.  
By Proposition \ref{prop:quartets}, 
$l(q)  =0$  if $q\in \calm_{T_{12|34}}$, while $l(q)
>0$ if $q\in \calm^+_{T_{13|24}}$ or $\calm^+_{T_{14|23}}$. By the linearity of $l$, this implies
$l(q)=0$ if $q\in \calm_{\cals}$, while $l(q)>0$ for $q\in\calm_{\calt}^{++}$ provided $\calt$ contains at least one of the resolved trees
$T_{13|24}$ or $T_{14|23}$.
 This implies that if $q \in \calm_{\cals} \cap \calm_{\calt}^{++}$, then no quartet incompatible with tree $T_{1}$ can appear among the trees of $\calt$.

If $T_{1}$ is the star tree, then from each of its three resolutions we obtain inequalities analogous to those for $l(q)$.  These imply that $\calt$ can only contain star trees.

\smallskip

Finally, in the case that all $T\in \calt$ are binary and $T_1\notin\calt$, if $q\in \calm_\cals \cap\calm_\calt^+$ then by replacing $\calt$ by a subset $\calt'$ we have $q\in \calm_\cals\cap\calm_{\calt'}^{++}$. From the argument above it follows that for all $T\in \calt'$ and all quartets $K$,
$T|_K=T_1|_K$. Thus we obtain the contradiction that $T=T_1$, and conclude no such $q$ exists.
\end{proof}

\medskip

\begin{proof}[Proof of Theorem \ref{thm:contract}]
We first consider the case that $T$ is
a 3-leaf  tree, and $T_{1}$ and $T_{2}$ are two of its contractions where one leaf has become an internal vertex.

The model on a 3-leaf tree under the JC model has precisely
$3$ nontrivial Fourier parameters, one per edge.  We set the parameterization of that
model, with edge parameters $a,b,c\in (0,1]$,  equal to the one for the mixture on $T_{1}$ and $T_{2}$, with edge parameters $d,e$ and $f,g$ respectively, and mixing parameter $\pi$.
This gives us, for fixed $a,b,c$, the following system of $4$ equations in
$5$ unknowns:
\begin{eqnarray*}
ab & = & (1-\pi)d + \pi f , \\
ac & = & (1-\pi)de + \pi g,  \\
bc & = & (1-\pi)e +  \pi fg,  \\
abc & = & (1-\pi)de + \pi fg.
\end{eqnarray*}
It is not difficult to see that the values
\begin{equation}\label{eq:plugin}
d = 0,\  e = bc,\   f = b,\  g = c,\  \pi = a
\end{equation}
give a solution to this system 
of equations.  
For the open models, however, we seek solutions where $d,e,f,g,\pi \in (0,1)$ for fixed $a,b,c \in (0,1)$.
A computation of the Jacobian of the system of equations at the values in equations \eqref{eq:plugin} 
allows us to apply the implicit function theorem, and treat $d$ as an independent variable in a neighborhood of the above solution.  Hence,
if we perturb $d$ to $d'$, with $0<  d' \ll 1$, we obtain 
parameters in $(0,1)$ solving the system of equations.  
This shows that there is complete mimicking for the open models in the 3-leaf case.

Finally we apply Proposition \ref{prop:fusion2}: Since any trees of the type specified
in the statement of the theorem can be obtained by attaching fusion ends to
the 3-leaf tree and its two degenerations, we deduce the general result.
\end{proof}


\end{document}